\documentclass[a4paper,12pt]{article}

\usepackage{amsmath}
\usepackage{amssymb}
\usepackage{amsthm}

\newcommand{\prt}{\partial}
\def\Re{{\mathbb R}}

\def\Pr{{\mathbb P}}
\def\H{{\mathcal H}}
\def\B{{\mathcal B}}

\def\L{{\hat L}}

\def\u{{\bar u}}
\def\w{{\bar w}}
\newtheorem{theorem}{Theorem}
\newtheorem{prop}[theorem]{Proposition}
\newtheorem{lemma}{Lemma}

\theoremstyle{definition}

\begin{document}
\bibliographystyle{plain}

\title{{\Large\bf  Universal Structure and Universal PDE for Unitary Ensembles}}
\author{Igor Rumanov\footnote{e-mail: igorrumanov@math.ucdavis.edu} \\
{\small Department of Mathematics, UC Davis, 1 Shields Avenue, CA 95616}}

\maketitle

\bigskip
\begin{abstract}
Random matrix ensembles with unitary invariance of measure (UE) are described in a unified way, using a combination of Tracy-Widom (TW) and Adler-Shiota-Van Moerbeke (ASvM) approaches to derivation of partial differential equations (PDE) for spectral gap probabilities. First, general 3-term recurrence relations for UE restricted to subsets of real line, or, in other words, for functions in the resolvent kernel, are obtained. Using them, simple universal relations between all TW dependent variables and one-dimensional Toda lattice $\tau$-functions are found. A universal system of PDE for UE is derived from previous relations, which leads also to a {\it single independent PDE} for spectral gap probability of various UE. Thus, orthogonal function bases and Toda lattice are seen at the core of correspondence of different approaches. Moreover, Toda-AKNS system provides a common structure of PDE for unitary ensembles. Interestingly, this structure can be seen in two very different forms: one arises from orthogonal functions-Toda lattice considerations, while the other comes from Schlesinger equations for isomonodromic deformations and their relation with TW equations. The simple example of Gaussian matrices most neatly exposes this structure.
\end{abstract}

\newpage 

\section*{\normalsize\bf I. INTRODUCTION}

Unitary ensembles of (single) random matrices (UE) are the most well known and thoroughly studied among random matrix ensembles (RME). Integrable properties possessed by their correlation functions and spectral gap probabilities, their ties with integrable hierarchies (see e.g. Ref.~\cite{DJKM}) of partial differential equations (PDE) were extensively studied in the last two decades, see e.g. Ref.~\cite{Me04}. Still the approaches used to derive equations satisfied by gap probabilities, e.g. Refs.~\cite{TW1, ASvM, BorDei}, seem to be far from each other, moreover, the dissimilarities grow when considering more complicated RME like coupled or Pfaffian ensembles. Due to the aforementioned integrability, there are some PDE in every approach, which are universal, i.e. independent of the particular UE being studied. They may be equations satisfied by resolvent kernels of Fredholm integral operators and functions they are composed of. This was established as a general approach to study RME spectra in Ref.~\cite{TW1} (we call it TW further on). They also may be well-known nonlinear integrable PDE --- the first members of KP or Toda lattice hierarchies, or their subhierarchies like KdV and others. This is the case in the approach started in Ref.~\cite{ASvM} (called ASvM from now on), see also Refs.~\cite{AvM2, AvM7, AvMdKP, PvM2007}. But even this universal part of all methods, supplied by the integrable structure itself, and thus identical for different ensembles, is described differently in various approaches. As for the non-universal part, depending on the particular properties of a given probability measure, it practically leads to the necessity of almost case by case study of various ensembles. However, one might hope that it should be possible after all to describe the universal properties in a universal way, incorporating all existing approaches and emphasizing the common underlying structure inherent in all cases. 
\par First steps in this direction were made in Refs.~\cite{Har} and \cite{IR1}, where the important TW and ASvM approaches were matched for several interesting cases: Airy, Bessel and Sine in Ref.~\cite{Har} and Gaussian (or Hermite) in Ref.~\cite{IR1}. In the second paper some very simple relations between the dependent variables of TW and ASvM approaches for the Gaussian UE (GUE) were revealed. Here we extend these relations to all UE with spectrum on real line and expose some additional relations of the same kind, which seem to be overlooked in all previous considerations. (Everything we do here can be repeated for circular UE or even for orthogonal-polynomial ensembles with eigenvalues on a contour in complex plane, considered e.g. in Ref.~\cite{BeEyHa-06}. Then, however, some specific formulas will differ from presented here, because of a different form of the basic three-term recurrence relations -- our starting point, see below.)  This allows us to derive some universal forms of PDE satisfied by gap probabilities of all UE irrespective of the specific spectral measure (or potential, in other words). Besides, we make contact with the isomonodromic approach~\cite{JMMS, JMU, JM2, IIKS, BorDei, BeEyHa-06} first applied to random matrices in Refs.~\cite{JMMS, FIK1, FIK2, HarTW}. In connection with TW work, it was most clearly exposed by Palmer in Ref.~\cite{Pal}, where the generalized Schlesinger isomonodromic deformation equations for UE were written with a common form of non-universal, i.e. potential dependent, terms. We demonstrate that combinations of these last equations in general give the form of ASvM equations for $\tau$-functions obtained in Ref.~\cite{IR1} for the Gaussian case. These combinations expose an amazing similarity in structure with our universal PDE claimed above. At the core of all our considerations is still, as it was in Ref.~\cite{IR1}, the structure of orthogonal function bases, three-term recurrence relations they satisfy, and one-dimensional Toda lattice hierarchy whose flows are in fact the continuous transformations among different such bases. We would like to stress here that while the last sentence is, strictly speaking, directly related only to {\it finite size} UE, the infinite ensembles can always be obtained as limits of finite ones. Moreover, the simple relations between TW and ASvM dependent variables, revealed here, allow one to {\it define} some quantities, originally meaningful only for finite matrices, like the $\tau$-ratios of matrix integrals of different matrix dimension, also for infinite ensembles.
\par The plan of the paper is as follows. Section II gives a brief summary of the key results in the paper. In section III we derive the three-term recurrence relations for resolvent kernel related orthogonal functions. The functions are identified with orthogonal functions for ensembles with measure restricted to the complement of the original measure support~\cite{BorSosh}. Using these recurrences, in section IV we derive the universal relations between the TW and ASvM variables extending and generalizing results of Ref.~\cite{IR1} for the Gaussian case to all UE. In section V, from the results of previous sections, we derive new universal PDE satisfied by all UE. We obtain three equations reflecting the fundamental Toda-AKNS structure but in addition also another equation arising from the defining connection between logarithmic derivative of $\tau$-function and resolvent kernel. This allows us to obtain a single universal PDE for gap probabilities of UE modified by Toda times in ASvM approach, separately for each endpoint of the spectrum. Besides, we obtain a new interesting recursion structure for finite size matrix integrals. In the next section we introduce a different framework for UE --- that of the isomonodromic deformations and Schlesinger equations, following the important work of Palmer~\cite{Pal}. We show then that our Toda-AKNS type system, derived by ASvM method in Ref.~\cite{IR1}, can always be obtained also by summing up the Schlesinger equations. A partial summation of them gives also a Toda-AKNS-type system analogous to the one from previous section, this time, however, describing {\it unmodified} UE, i.e. ensembles with fixed couplings in the potential. The additional equation of previous section is now a consequence of summed Schlesinger system. Another consequence is two universal constraints for the non-universal terms in TW/Schlesinger equations. In section VII we consider simplest examples from the perspective of universal PDE (\ref{eq:T}) of section V and explain the difficulty of its direct application to more complex cases. Section VIII is devoted to conclusions and open questions. 

\section*{\normalsize\bf II. BRIEF SUMMARY OF MAIN RESULTS}

For any UE on real line, given the corresponding orthogonal polynomial system and related Christoffel-Darboux kernel $K_n$, functions $Q_n(x)$, $P_n(x)$ ($n$ -- matrix size), defining the resolvent kernel $R_n$ of its restriction $K_n^J$ to a subset $J^c$ of $\Re$ also satisfy certain three-term recurrence relations,

\begin{equation}
xQ_n = b_nQ_{n+1}+(a_n+v_n-v_{n+1})Q_n+b_{n-1}(1-\bar u_n)(1+\bar w_n)Q_{n-1},     \label{eq:3Q}  
\end{equation}


\begin{equation}
xP_n = b_{n-1}(1-\bar u_n)(1+\bar w_n)P_{n+1}+(a_{n-1}-v_n+v_{n-1})P_n+b_{n-2}P_{n-1},    \label{eq:3P}   
\end{equation}

\noindent with coefficients explicitly determined by inner products $u_n$, $v_n$ and $w_n$ (\ref{eq:11}) over the subset, which were the additional auxiliary variables introduced by Tracy and Widom in Refs.~\cite{TW-Airy, TW1} to derive PDE for level spacing probabilities. Here $\bar u_n = u_n/b_{n-1}$, $\bar w_n = w_n/b_{n-1}$, and $a_n$, $b_n$ are the coefficients in the original three-term relations (\ref{eq:1}). There are also simple recursions among the above quantities:

\begin{equation}
P_{n+1}(x) = \frac{Q_n(x)}{1-\bar u_n}, \hspace{1.5cm} Q_{n-1}(x) = \frac{P_n(x)}{1+\bar w_n}, \label{eq:16} 
\end{equation}

\begin{equation}
\bar w_{n+1} = \bar u_n/(1-\bar u_n). \label{eq:17} 
\end{equation}

\noindent Renormalized functions

\begin{equation}
r_k(x) = \frac{Q_k(x)}{(1-\bar u_k)^{1/2}},   \label{eq:21} 
\end{equation}

\noindent are  orthonormal w.r.t. the original measure $\exp(-V(x))dx$ restricted to the complement $J$ of the subset. Thus, they also satisfy a usual Christoffel-Darboux formula:

\begin{equation}
R_n(x,y) = \sum_{k=0}^{n-1}r_k(x)r_k(y) = b_{n-1}(1-\bar u_n)^{1/2}(1+\bar w_n)^{1/2}\frac{r_n(x)r_{n-1}(y)-r_{n-1}(x)r_n(y)}{x-y},  \label{eq:22} 
\end{equation}

\noindent and the three-term relations with symmetric (Jacobi) matrix of coefficients

$$
xr_n(x) = b_n(1-\bar u_{n+1})^{1/2}(1+\bar w_{n+1})^{1/2}r_{n+1}(x)+(a_n+v_n-v_{n+1})r_n(x)+
$$

\begin{equation}
+ b_{n-1}(1-\bar u_n)^{1/2}(1+\bar w_n)^{1/2}r_{n-1}(x).  \label{eq:23} 
\end{equation}

\noindent The above formulas add in fact some more specific details to the earlier result of Borodin and Soshnikov~\cite{BorSosh} on restricted RME.
\par As a consequence of the above simple facts and general connection of orthogonal polynomial systems with Toda lattice, the simple universal relations hold with the UE matrix integrals, which are Toda $\tau$-functions:

\begin{equation}
u_n \equiv (\varphi, (I-K_n^J)^{-1}\varphi) = b_{n-1}\left(1-\frac{\tau_{n+1}^J/\tau_{n+1}}{\tau_n^J/\tau_n}\right) = \sqrt\frac{\tau_{n+1}\tau_{n-1}}{(\tau_n)^2}\left(1-\frac{\tau_{n+1}^J/\tau_{n+1}}{\tau_n^J/\tau_n}\right),   \label{eq:30} 
\end{equation}

\begin{equation}
w_n \equiv (\psi, (I-K_n^J)^{-1}\psi) = b_{n-1}\left(\frac{\tau_{n-1}^J/\tau_{n-1}}{\tau_n^J/\tau_n}-1\right) = \sqrt\frac{\tau_{n+1}\tau_{n-1}}{(\tau_n)^2}\left(\frac{\tau_{n-1}^J/\tau_{n-1}}{\tau_n^J/\tau_n}-1\right),   \label{eq:32} 
\end{equation}
 
\begin{equation}
\left.  v_n \equiv (\varphi, (I-K_n^J)^{-1}\psi) \equiv (\psi, (I-K_n^J)^{-1}\varphi) = -\frac{\prt}{\prt t_1}\ln\frac{\tau_n^J}{\tau_n}\right|_{t=0}. \label{eq:41a} 
\end{equation}

\noindent where $\tau_n^J$ is the matrix integral over $J^n$, $\tau_n$ -- the same integral over $\Re$, and $t_1$ is the first Toda time.
\par Due to the last relations, there is a universal system of PDE with respect to {\it each single} endpoint of the spectrum $\xi$ {\it and} $t_1$  satisfied by $T \equiv \ln\tau_n^J$ and $\tau$-ratios $U_n \equiv \frac{\tau_{n+1}^J}{\tau_n^J}$ and $W_n \equiv \frac{\tau_{n-1}^J}{\tau_n^J}$, if we consider the original measure modified by a linear potential, $V(x) \to V(x) - t_1x$:
 
\begin{equation}
\left(\frac{\prt\ln\tau_n^J}{\prt\xi}\right)^2 = R_n^2(\xi,\xi) = -\frac{1}{4}\frac{\prt U_n}{\prt\xi}\frac{\prt W_n}{\prt \xi}\left(\frac{\prt}{\prt \xi}\ln\left(-\frac{\prt U_n/\prt \xi}{\prt W_n/\prt \xi}\right)\right)^2.  \label{eq:51}  
\end{equation}

\begin{equation}
\frac{\prt^2 \ln\tau_n^J}{\prt t_1^2} = -\frac{\prt v_n}{\prt t_1} + \frac{\prt^2 \ln\tau_n}{\prt t_1^2} = U_nW_n,  \label{eq:52} 
\end{equation}

\begin{equation}
\frac{\prt^2 U_n}{\prt \xi\prt t_1} = \xi\frac{\prt U_n}{\prt\xi} + 2\frac{\prt v_n}{\prt\xi}U_n,  \label{eq:53} 
\end{equation}

\begin{equation}
\frac{\prt^2 W_n}{\prt \xi\prt t_1} = -\xi\frac{\prt W_n}{\prt\xi} + 2\frac{\prt v_n}{\prt\xi}W_n,  \label{eq:54} 
\end{equation}

\noindent Here (\ref{eq:52}) is the well-known Toda equation, and the other three PDE satisfied by any such modified UE can be written equivalently as ($t \equiv t_1$, $U = U_n$, $W = W_n$  below)

\begin{equation}
\prt_t(T_{\xi})^2 = U_{\xi}W_{\xi\xi} - W_{\xi}U_{\xi\xi}, \label{eq:57} 
\end{equation}

\begin{equation}
\left(T_{\xi t}\right)^2 = -U_{\xi}W_{\xi},  \label{eq:55} 
\end{equation}

\begin{equation}
WU_{\xi t} - UW_{\xi t}  = \xi T_{\xi tt}. \label{eq:56} 
\end{equation}

\noindent It is possible to eliminate $U$ and $W$ from (\ref{eq:57}), (\ref{eq:55}) and (\ref{eq:52}), and we obtain a {\it universal PDE for the logarithm of gap probability} $\Pr$ ($T = \ln\tau_n^J = \ln\Pr + \ln\tau_n$):

\begin{equation}
\left(T_{\xi t} T_{\xi\xi tt} - T_{\xi tt}T_{\xi\xi t} + 2(T_{\xi t})^3\right)^2 = (T_{\xi})^2\left(4T_{tt}(T_{\xi t})^2 + (T_{\xi tt})^2\right).  \label{eq:T}  
\end{equation}

\noindent In addition, using also (\ref{eq:56}) allows to derive a universal PDE for the logarithm of the ratio $T_+ \equiv U_n \equiv \ln(\tau_{n+1}^J/\tau_n^J)$ (where `prime' stands for the derivative w.r.t. $\xi$):

\begin{equation}
(T_+'')_{tt} - \frac{(T_+')_{tt}T_+''}{T_+'} - (T_+')^2(T_+)_{tt} - \frac{(T_+')_{t}(T_+'')_t}{T_+'} + \frac{(T_+')_{t}^2T_+''}{(T_+')^2} + (\xi - (T_+)_t)T_+'(2(T_+')_t - 1) = 0.  \label{eq:T+}  
\end{equation}
 
\par The universal Toda-AKNS-like structure of (\ref{eq:52})--(\ref{eq:54}) appears in a different guise from Schlesinger equations of isomonodromic approach, where Toda time $t_1$ is not involved, in section VI.

\section*{\normalsize\bf III. RESOLVENT KERNEL${\backslash}$RESTRICTED ENSEMBLES' 3-TERM RELATIONS}

Consider a system of orthogonal functions (quasi-polynomials) $\{\varphi_n\}$ that satisfy a 3-term relation ( $\varphi_n(x) = \pi_n(x)\exp(-V(x)/2)$, $\{\pi_n\}$ are polynomials orthonormal w.r.t. a probability measure $\exp(-V(x))dx$ on real line $\Re$ ):

\begin{equation}
x\varphi_n = b_n\varphi_{n+1} + a_n\varphi_n + b_{n-1}\varphi_{n-1}, \label{eq:1} 
\end{equation}

\noindent which can be written as 

\begin{equation}
x\varphi_n = (\L\varphi)_n, \hspace{1cm} \L = b\Lambda^T + a\Lambda^0 + b\Lambda, \label{eq:1a} 
\end{equation}

\noindent i.e. $\L$ - infinite matrix with $n$-th row $(\L_n)_j = b_n\delta_{n+1,j} + a_n\delta_{n,j} + b_{n-1}\delta_{n-1,j}$ ($\Lambda_{ij} = \delta_{i,j+1}$ - the infinite shift matrix, $a = diag(a_0, a_1, a_2, ...)$, $b = diag(b_0, b_1, b_2, ...)$). The corresponding Christoffel-Darboux formula reads:

\begin{equation}
K_n(x, y) = \frac{\varphi(x)\psi(y) - \psi(x)\varphi(y)}{x - y} = \sum_{k=0}^{n-1}\varphi_k(x)\varphi_k(y), \label{eq:2} 
\end{equation}

\noindent where $\varphi(x) = \sqrt{b_{n-1}}\varphi_n(x)$, $\psi(x) = \sqrt{b_{n-1}}\varphi_{n-1}(x)$. Denote by $K_n^J$ the Fredholm operator with kernel (\ref{eq:2}) acting {\it on the complement} $J^c$ of subset $J$ of $\Re$. The resolvent kernel (we will sometimes use short-hand notation $K$ for $K_n^J$), $R(x, y) \equiv R_n^J(x, y)$, the kernel of $K(I-K)^{-1}$, is defined so that the operator identity holds:

\begin{equation}
(I + R)(I - K) = I.   \label{eq:3} 
\end{equation}

Let us introduce the auxiliary functions playing a prominent role in TW approach~\cite{TW1}:

\begin{equation}
Q(x;J) = (I-K)^{-1}\varphi(x), \hspace{1cm} P(x;J) = (I-K)^{-1}\psi(x). \label{eq:4} 
\end{equation}

\noindent It is known since the seminal work~\cite{IIKS} that if $K$ is an {\it integrable} operator, i.e. its kernel can be expressed by the left-hand side of the Christoffel-Darboux formula (\ref{eq:2}), the resolvent operator is also integrable and its kernel can be written the same way in terms of the functions $Q$ and $P$:

\begin{equation}
R_n(x, y) = \frac{Q(x)P(y) - P(x)Q(y)}{x - y} = b_{n-1}\frac{Q_n(x)P_n(y) - P_n(x)Q_n(y)}{x - y}, \label{eq:5} 
\end{equation}

\noindent where we let $Q = \sqrt{b_{n-1}}Q_n$, $P = \sqrt{b_{n-1}}P_n$, then

\begin{equation}
\varphi_n = (I-K)Q_n, \hspace{2cm} \varphi_{n-1} = (I-K)P_n  \label{eq:6} 
\end{equation}

\noindent Important for our purposes will be inner products $u$, $v$ and $w$, introduced in Ref.~\cite{TW1}:

\begin{equation}
\begin{array}{l} u \equiv u_n = \int_{J^c} Q(x;J)\varphi(x)dx, \ \ \ \ \ w \equiv w_n = \int_{J^c} P(x;J)\psi(x)dx, \\  \\  v \equiv v_n = \int_{J^c} P(x;J)\varphi(x)dx = \int_{J^c} Q(x;J)\psi(x)dx \end{array}   \label{eq:11} 
\end{equation}

\noindent the second equality in the definition of $v$ being true by symmetry of the operator $K_n^J$ and, consequently, also $R_n^J$. Let $\bar u_n = u_n/b_{n-1}$, $\bar w_n = w_n/b_{n-1}$, then

\begin{lemma} 

$$
P_{n+1}(x) = \frac{Q_n(x)}{1-\bar u_n}, \hspace{1.5cm} Q_{n-1}(x) = \frac{P_n(x)}{1+\bar w_n}.   \eqno(\ref{eq:16}) 
$$

$$
\bar w_{n+1} = \bar u_n/(1-\bar u_n).   \eqno(\ref{eq:17}) 
$$

\end{lemma}

\begin{proof}
\par Functions $P_{n+1}$ and $Q_n$ (as well as $Q_{n-1}$ and $P_n$) are not independent:

\begin{equation}
\varphi_n = (I-K)Q_n = (I-K_{n+1})P_{n+1} = (I-K)P_{n+1} - (K_{n+1}-K)P_{n+1}, \label{eq:13} 
\end{equation}

\noindent which can be rewritten as

\begin{equation}
\left(I - (I-K)^{-1}(K_{n+1}-K)\right)P_{n+1} = Q_n. \label{eq:14}  
\end{equation}

\noindent On the other hand, the operator difference $K_{n+1}^J-K_n^J$ in (\ref{eq:14}) is a simple projector (with $\dot=$ meaning the kernel of the operator on the left):

\begin{equation}
K_{n+1}-K\ {\dot=}\ \varphi_n(x)\varphi_n(y)\chi_{J^c}(y), \label{eq:15}  
\end{equation}

\noindent where $\chi_{J^c}(y)$ is the characteristic functon of $J^c$, equal to $1$ on $J^c$ and $0$ outside. After applying (\ref{eq:15}) in (\ref{eq:14}) and using the definition of quantities $u_n$ and $w_n$, we find simple recursion formulas for $P_{n+1}$ and $Q_{n-1}$ in terms of $Q_n$ and $P_n$, i.e. (\ref{eq:16}). Shifting indices in (\ref{eq:16}), we find (\ref{eq:17}) as their consistency condition.

\end{proof}

\begin{theorem}

Functions $Q_n$ and $P_n$ in the resolvent kernel of a Fredholm operator associated with finite size UE, satisfy three term recurrence relations:

$$
xQ_n = b_nQ_{n+1}+(a_n+v_n-v_{n+1})Q_n+b_{n-1}(1-\bar u_n)(1+\bar w_n)Q_{n-1},  \eqno(\ref{eq:3Q})   
$$

$$
xP_n = b_{n-1}(1-\bar u_n)(1+\bar w_n)P_{n+1}+(a_{n-1}-v_n+v_{n-1})P_n+b_{n-2}P_{n-1}.  \eqno(\ref{eq:3P})   
$$

\end{theorem}

\begin{proof}
\noindent From (\ref{eq:1a}) and (\ref{eq:6}) we have

\begin{equation}
x\varphi_n = (\L(I-K)Q)_n = (I-K)(\L Q)_n - ([\L,K]Q)_n, \label{eq:7} 
\end{equation}

\noindent where the commutator gives

$$
([\L,K]Q)_n = (\L KQ)_n - K(\L Q)_n = 
$$

$$
= (b_nK_{n+1}Q_{n+1}+a_nKQ_n+b_{n-1}K_{n-1}Q_{n-1}) - K(b_nQ_{n+1}+a_nQ_n+b_{n-1}Q_{n-1}) =   
$$

\begin{equation}
= b_n(K_{n+1}-K)Q_{n+1} + b_{n-1}(K_{n-1}-K)Q_{n-1}.  \label{eq:8} 
\end{equation}

\noindent Similarly, 

$$
x\varphi_{n-1} = (I-K)(\L P)_n - ([\L,K]P)_n \eqno(\ref{eq:7}a) 
$$

\noindent with

$$
([\L,K]P)_n = b_{n-1}(K_{n+1}-K)P_{n+1} + b_{n-2}(K_{n-1}-K)P_{n-1}. \eqno(\ref{eq:8}a) 
$$

\noindent Applying operator $(I-K)^{-1}$ to (\ref{eq:7}), (\ref{eq:7}a) we get

\begin{equation}
(I-K)^{-1}x\varphi_n = (\L Q)_n - (I-K)^{-1}([\L,K]Q)_n,  \label{eq:9} 
\end{equation}

$$
(I-K)^{-1}x\varphi_{n-1} = (\L P)_n - (I-K)^{-1}([\L,K]P)_n  \eqno(\ref{eq:9}a) 
$$

\par For the left-hand side of (\ref{eq:9}) we derive (imitating similar derivations of Refs.~\cite{IIKS, TW1}, see also Ref.~\cite{BIK}):

$$
(I-K)^{-1}x\varphi_n = \frac{1}{\sqrt{b_{n-1}}}(I-K)^{-1}x\varphi = \frac{1}{\sqrt{b_{n-1}}}(I+R)x\varphi = 
$$

$$
= \frac{1}{\sqrt{b_{n-1}}}\int_J\left(\delta(x-y)+\frac{Q(x)P(y)-Q(y)P(x)}{x-y}\right)\left(x\varphi(y) + (y-x)\varphi(y)\right)dy =   
$$

\begin{equation}
= \frac{1}{\sqrt{b_{n-1}}}\left(x(I+R)\varphi - \int_J(Q(x)P(y)-Q(y)P(x))\varphi(y)dy\right) = xQ_n - vQ_n + uP_n, \label{eq:10} 
\end{equation}

\noindent Completely analogously we get

\begin{equation}
(I-K)^{-1}x\varphi_{n-1} = xP_n - wQ_n + vP_n. \label{eq:12} 
\end{equation}

\noindent On the right-hand side of (\ref{eq:9}) and (\ref{eq:9}a), respectively, using (\ref{eq:8}), (\ref{eq:8}a), (\ref{eq:15}) and the definitions (\ref{eq:11}), we get

\begin{equation}
(I-K)^{-1}([\L,K]Q)_n = v_{n+1}Q_n - b_{n-1}\u_{n-1}P_n, \label{eq:18} 
\end{equation}

$$
(I-K)^{-1}([\L,K]P)_n = b_{n-1}\w_{n+1}Q_n - v_{n-1}P_n. \eqno(\ref{eq:18}a) 
$$

\noindent Finally, plugging (\ref{eq:10}), (\ref{eq:18}) into (\ref{eq:9}), and (\ref{eq:12}), (\ref{eq:18}a) into (\ref{eq:9}a), we come to the three-term relations (\ref{eq:3Q}), (\ref{eq:3P}). 
\end{proof}

\par Consider now the difference of resolvent kernels for the operators of consecutive rank:

$$
R_{n+1}(x,y)-R_n(x,y) = \frac{Q_n(x)Q_n(y)}{1-\bar u_n} = Q_n(x)P_{n+1}(y) = P_{n+1}(x)Q_n(y),
$$

\noindent therefore

$$
R_n(x,y) = b_{n-1}\frac{Q_n(x)P_n(y)-P_n(x)Q_n(y)}{x-y} = \sum_{k=0}^{n-1}Q_k(x)P_{k+1}(y).
$$

\noindent If we introduce the following normalization for the orthogonal functions:

$$
r_k(x) = \frac{Q_k(x)}{(1-\bar u_k)^{1/2}}, \eqno(\ref{eq:21})  
$$

\noindent then one can easily check that the new functions $r_k$ are  orthonormal w.r.t. the measure $\exp(-V(x))dx$ restricted to $J$, i.e. the original measure on $J$ and zero outside of it. Thus, in terms of these functions we can write out 

\begin{theorem}

The Christoffel-Darboux formula for the resolvent kernel holds:

$$
R_n(x,y) = \sum_{k=0}^{n-1}r_k(x)r_k(y) = b_{n-1}(1-\bar u_n)^{1/2}(1+\bar w_n)^{1/2}\frac{r_n(x)r_{n-1}(y)-r_{n-1}(x)r_n(y)}{x-y}, \eqno(\ref{eq:22}) 
$$

\noindent and the three-term relations with symmetric (Jacobi) matrix of coefficients are

$$
xr_n(x) = b_n(1-\bar u_{n+1})^{1/2}(1+\bar w_{n+1})^{1/2}r_{n+1}(x)+(a_n+v_n-v_{n+1})r_n(x)+
$$
                    
$$
+ b_{n-1}(1-\bar u_n)^{1/2}(1+\bar w_n)^{1/2}r_{n-1}(x). \eqno(\ref{eq:23}) 
$$

\end{theorem}

At this point one should recall the contents of the paper by Borodin and Soshnikov~\cite{BorSosh}. There the setup is completely identical to ours: two complementary subsets $J$ and $J^c = \Re\setminus J$ of $\Re$ are considered, and the corresponding operator defined by (3). Just the name ``resolvent kernel" is not used, another name -- ``Janossy densities" -- is applied. The authors, however, in fact give a simple proof of orthogonality of the above functions $Q$, $P$ or $r$, so we just restate it briefly here in terms of Refs.~\cite{IIKS, TW1}. By definition of functions $Q_n$ we can write

$$
\varphi_n = (I - K_n^J)Q_n = Q_n - \sum_{k=0}^{n-1}\varphi_k(\varphi_k, Q_n)_{J^c},
$$

\noindent or

\begin{equation}
Q_n = \varphi_n + \sum_{k=0}^{n-1}A_{nk}\varphi_k,  \label{eq:24} 
\end{equation}

\noindent where $A_{nk} = (\varphi_k, Q_n)_{J^c}$ -- the (matrix of) inner products, so that $A_{nn} = \bar u_n$. Taking inner products of (\ref{eq:24}) with $\varphi_k$, $k  \le n$, over $\Re$, we see that

$$
(\varphi_k, Q_n)_{\Re} = (\varphi_k, Q_n)_{J^c}, \text{ for } k < n,
$$

\noindent so $\{Q_n\}$ are orthogonal on $J$ and

$$
(\varphi_n, Q_n)_{\Re} = (\varphi_n, \varphi_n)_{\Re} = 1.
$$

\noindent It follows that

$$
(Q_n, Q_n)_{J} = (\varphi_n, Q_n)_{J} = (\varphi_n, Q_n)_{\Re} - (\varphi_n, Q_n)_{J^c} = 1- \u_n,
$$

\noindent which shows that functions $\{r_n(x; J)\}$ defined by (\ref{eq:21}) are indeed orthonormal on $J \subset \Re$. Following Ref.~\cite{BorSosh}, let us introduce the $n \times n$ matrices with elements ($j, k  = 0,1,\dots,n-1$) 

$$
(G_J)_{jk} = \int_J \varphi_j(x)\varphi_k(x)dx, \ \ \  (G_{J^c})_{jk} = \delta_{jk}-(G_J)_{jk} = \int_{J^c} \varphi_j(x)\varphi_k(x)dx
$$

\noindent Taking inner products of (\ref{eq:24}) with $\varphi_k$ over $J^c$, one can easily see that

$$
\delta_{jk} + A_{jk} = \left(\delta_{jk} - (G_{J^c})_{jk}\right)^{-1} = \left((G_J)_{jk}\right)^{-1},
$$

\noindent i.e. the relation for the restriction of $R_n^J$ to the $n$-dimensional subspace $\H = \text{span} \{\varphi_j\}_{j=0,\dots,n-1}$ equal to $G_{J^c}(I-G_{J^c})^{-1}$, which was crucial in the argument of Ref.~\cite{BorSosh}. \\
{\it Remark.} Also, see Proposition 2.5 of Ref.~\cite{Bor}, if we take the Gauss decomposition of $G_J$, $G_J = S^{-1}(S^T)^{-1}$, where $S$ is lower-triangular, then matrix $S$ gives the coefficients of expansion of $r_k$ w.r.t. the $\varphi_j$ basis: $r_k(x; J) = \sum_{j=0}^{n-1}S_{kj}(J)\varphi_j(x)$, $k = 0, \dots, n-1$.
\par A more recent consideration of Janossy densities for general UE at the spectral edge, using Riemann-Hilbert problem techniques was undertaken in Ref.~\cite{RiZh}, also without making explicit connections with the framework of Ref.~\cite{TW1}.

\section*{\normalsize\bf IV. UNIVERSAL RELATIONS BETWEEN TW VARIABLES AND TODA $\tau$-FUNCTIONS}

By definitions of the matrix integral over $J^n$, $\tau_n^J$, and the resolvent operator $R_n^J = K_n^J(I-K_n^J)^{-1}$, we have the fundamental relation between them:

\begin{equation}
R_n^J(a_j,a_j) = (-1)^{j-1}\frac{\prt\ln\tau_n^J}{\prt a_j}.  \label{eq:25} 
\end{equation}

\noindent This is just a consequence of the defining connection, the equality of two different expressions for the probability of all eigenvalues to lie in $J$,

\begin{equation}
\det(I - K_n^J) = \frac{\tau_n^J}{\tau_n},  \label{eq:26} 
\end{equation}

\noindent where $\tau_n$ is the corresponding matrix integral over whole $\Re^n$. Moreover, we have

\begin{theorem}

$$
u_n \equiv (\varphi, (I-K_n^J)^{-1}\varphi) = b_{n-1}\left(1-\frac{\tau_{n+1}^J/\tau_{n+1}}{\tau_n^J/\tau_n}\right) = \sqrt\frac{\tau_{n+1}\tau_{n-1}}{(\tau_n)^2}\left(1-\frac{\tau_{n+1}^J/\tau_{n+1}}{\tau_n^J/\tau_n}\right).  \eqno(\ref{eq:30}) 
$$

$$
w_n \equiv (\psi, (I-K_n^J)^{-1}\psi) = b_{n-1}\left(\frac{\tau_{n-1}^J/\tau_{n-1}}{\tau_n^J/\tau_n}-1\right) = \sqrt\frac{\tau_{n+1}\tau_{n-1}}{(\tau_n)^2}\left(\frac{\tau_{n-1}^J/\tau_{n-1}}{\tau_n^J/\tau_n}-1\right).  \eqno(\ref{eq:32}) 
$$

\end{theorem}

\begin{proof}
Let us consider the values of our orthonormal functions $r_n(x)$ at the endpoints of $J$, $a_j$. From the expressions preceding the definition (\ref{eq:21}) of functions $r_n$ and (\ref{eq:25}) it follows that

\begin{equation}
r_n(a_j)^2 = R_{n+1}^J(a_j,a_j)-R_n^J(a_j,a_j) = (-1)^{j-1}\frac{\prt\ln(\tau_{n+1}^J/\tau_n^J)}{\prt a_j}. \label{eq:27} 
\end{equation}

\noindent On the other hand, let us recall the universal TW equations~\cite{TW1}

\begin{equation}
\frac{\prt u}{\prt a_j} = (-1)^j Q(a_j; J)^2,  \label{eq:28} 
\end{equation}

\begin{equation}
\frac{\prt w}{\prt a_j} = (-1)^j P(a_j; J)^2, \label{eq:31} 
\end{equation}

\begin{equation}
\frac{\prt v}{\prt a_j} = (-1)^j Q(a_j;J)P(a_j; J).  \label{eq:33} 
\end{equation}

\noindent From (\ref{eq:28}) and (\ref{eq:21}) we have also 

\begin{equation}
r_n(a_j)^2 = (-1)^{j-1}\frac{\prt\ln(1-\bar u_n)}{\prt a_j}. \label{eq:29} 
\end{equation}

\noindent Comparing the expressions (\ref{eq:27}) and (\ref{eq:29}), we get the simple universal relation (\ref{eq:30}) between $\bar u_n$ and $\tau$-ratios. The other universal relation, for $\bar w_n$ and corresponding $\tau$-ratios is obtained completely similarly, using (\ref{eq:31}) and Lemma 1.
\end{proof}
 
\noindent In fact, if we recall the recurrence relation between $\bar w_n$ and $\bar u_{n-1}$ (\ref{eq:17}), we see that (\ref{eq:32}) is the same relation as (\ref{eq:30}), just written in terms of $\bar w$ instead of $\bar u$. The last two formulas were derived in Ref.~\cite{IR1} for the Gaussian matrices only, using the specific properties of this case. Now it is clear that they hold for all UE. 
\par For the Gaussian case we also obtained the following relation~\cite{IR1}:

\begin{equation}
2v \equiv 2(\varphi, (I-K_n^J)^{-1}\psi) = \B_{-1}\ln\tau_n^J.  \label{eq:34} 
\end{equation}

\noindent This relation, as we will see, is not true for general UE. However, it turns out that it has a universal analogue. To see it, the framework of Toda lattice for UE, established in ASvM approach~\cite{ASvM, AvM7}, see also Ref.~\cite{PvM2007}, sections 5--8, is crucial. Following these works, let us consider the UE with modified probability measure:  $\rho(x) = e^{-V(x)} \to \rho_t(x) = e^{-V(x)+\sum_{k=1}^{\infty}t_kx^k}$, $t = (t_1, t_2, t_3, ...)$, then the matrix integral

\begin{equation}
\tau_n^J(t) = \frac{1}{n!}\int_{J^n} \Delta_n^2(x) \prod_1^n \rho_t(x_i) dx_i   \label{eq:35} 
\end{equation}

\noindent is a $\tau$-function of integrable hierarchies -- KP and 1-dimensional Toda lattice -- as is the corresponding integral over whole $\Re^n$, $\tau_n(t)$. The orthogonal functions for $t$-deformed weight satisfy $t$-dependent 3-term recurrence relation

$$
x\varphi_n(x;t) = b_n(t)\varphi_{n+1}(x;t) + a_n(t)\varphi_n(x;t) + b_{n-1}(t)\varphi_{n-1}(x;t),
$$

\noindent and the coefficients are expressed in terms of the Toda $\tau$-functions:

\begin{equation}
a_n(t) = \frac{\prt}{\prt t_1}\ln\frac{\tau_{n+1}(t)}{\tau_n(t)}, \label{eq:36} 
\end{equation}

\begin{equation}
b_{n-1}^2(t) = \frac{\tau_{n+1}(t)\tau_{n-1}(t)}{\tau_n^2(t)}.  \label{eq:37} 
\end{equation}

\noindent Let us denote $v_+ = v_{n+1}-v_n-a_n$, $v_- = v_{n-1}-v_n+a_{n-1}$. Then 3-term relation (\ref{eq:23}) together with the same Toda lattice correspondence for restricted (``resolvent") ensemble means:

\begin{equation}
b_{n-1}^2 = \frac{\tau_{n+1}\tau_{n-1}}{\tau_n^2} \rightarrow b_{n-1}^2(1-\bar u_n)(1+\bar w_n) = \frac{\tau_{n+1}^J\tau_{n-1}^J}{(\tau_n^J)^2}, \label{eq:38} 
\end{equation}

\begin{equation}
\left. a_n = \frac{\prt}{\prt t_1}\ln\frac{\tau_{n+1}}{\tau_n}\right|_{t=0}, \hspace{0.5cm} a_n \rightarrow -v_+, \label{eq:39} 
\end{equation}

\begin{equation}
\left.  v_+ = -\frac{\prt}{\prt t_1}\ln\frac{\tau_{n+1}^J}{\tau_n^J}\right|_{t=0}, \hspace{0.5cm} \left.  v_- = -\frac{\prt}{\prt t_1}\ln\frac{\tau_{n-1}^J}{\tau_n^J}\right|_{t=0}.  \label{eq:40} 
\end{equation}

\noindent From the last relations one can immediately conclude that 

\begin{theorem}

The inner product $v_n(t)$ can be expressed as

\begin{equation}
v_n(t) = -\frac{\prt}{\prt t_1}\ln\frac{\tau_n^J(t)}{\tau_n(t)}. \label{eq:41} 
\end{equation}

\noindent and the universal relation for the original (unmodified) ensemble follows by taking the point $t=0$:

$$
\left.  v_n = -\frac{\prt}{\prt t_1}\ln\frac{\tau_n^J}{\tau_n}\right|_{t=0}. \eqno(\ref{eq:41a}) 
$$

\end{theorem}
 
\noindent Now, if we continue to consider the $t$-modified RME, the above universal relations together with 3-term recursion relations for the restricted measure (or resolvent kernel) ensembles will lead us to some universal PDE for matrix integral $\tau$-functions, containing derivatives w.r.t. the endpoints {\it and} the first Toda time derivatives.

\section*{\normalsize\bf V. UNIVERSAL PDE FROM 3-TERM RELATIONS FOR RESOLVENT KERNEL}

Let us consider our 3-term relations from section 2 taken at a spectral endpoint $\xi$. Unlike the notations in Ref.~\cite{TW1} we denote $q_n = Q(\xi; J)$, $p_n = P(\xi; J)$, i.e. the subscript refers to the size of the matrix rather than to the order number of the endpoint as in Ref.~\cite{TW1}. Also let $\bar q_n = Q_n(\xi; J)$, $\bar p_n = P_n(\xi; J)$ (see (\ref{eq:5}), (\ref{eq:6})). Then we have from 3-term relations (\ref{eq:3Q}), (\ref{eq:3P}):

\begin{equation}
b_n\bar q_{n+1} = (\xi - a_n + v_{n+1} - v_n)\bar q_n - b_{n-1}(1-\bar u_n)\bar p_n,  \label{eq:42} 
\end{equation}

\begin{equation}
b_{n-2}\bar p_{n-1} = (\xi - a_{n-1} + v_n - v_{n-1})\bar p_n - b_{n-1}(1+\bar w_n)\bar q_n.  \label{eq:43}  
\end{equation}

\begin{theorem}

For all finite size UE, ratios of matrix integrals  $U_n \equiv \tau_{n+1}^J/\tau_n^J$, $W_n \equiv \tau_{n-1}^J/\tau_n^J$ and function $v_n(t) = -\prt t_1\ln(\tau_n^J(t)/\tau_n(t))$ satisfy the following PDE w.r.t. any specific spectral endpoint $\xi \in \prt J$ and the first 1-Toda time $t_1$:

$$
\frac{\prt^2 U_n}{\prt \xi\prt t_1} = \xi\frac{\prt U_n}{\prt\xi} + 2\frac{\prt v_n}{\prt\xi}U_n, \eqno(\ref{eq:53}) 
$$

$$
\frac{\prt^2 W_n}{\prt \xi\prt t_1} = -\xi\frac{\prt W_n}{\prt\xi} + 2\frac{\prt v_n}{\prt\xi}W_n.  \eqno(\ref{eq:54})  
$$

\end{theorem}

\begin{proof}
\noindent Using formulas (\ref{eq:16}) and (\ref{eq:33}), one gets for $\xi = a_k$:

\begin{equation}
\frac{\prt v_{n+1}}{\prt a_k} = (-1)^k b_n\bar q_{n+1}\bar p_{n+1} = (-1)^k b_n\bar q_{n+1}\frac{\bar q_n}{1-\bar u_n},  \label{eq:44} 
\end{equation}

\begin{equation}
\frac{\prt v_{n-1}}{\prt a_k} = (-1)^k b_{n-2}\bar q_{n-1}\bar p_{n-1} = (-1)^k b_{n-2}\bar p_{n-1}\frac{\bar p_n}{1+\bar w_n},  \label{eq:45} 
\end{equation}

\noindent Substituting (\ref{eq:42}) and (\ref{eq:43}) into (\ref{eq:44}) and (\ref{eq:45}), respectively, we get

$$
\frac{\prt v_{n+1}}{\prt \xi} = -(-1)^k b_{n-1}\bar q_n\bar p_n + (\xi - a_n + v_{n+1} - v_n)(-1)^k \frac{\bar q_n^2}{1-\bar u_n} =
$$

\begin{equation}
= -\frac{\prt v_n}{\prt\xi} - (\xi - a_n + v_{n+1} - v_n)\frac{\prt\ln(1-\bar u_n)}{\prt\xi},  \label{eq:46} 
\end{equation}

\begin{equation}
\frac{\prt v_{n-1}}{\prt \xi} = -\frac{\prt v_n}{\prt\xi} + (\xi - a_{n-1} + v_n - v_{n-1})\frac{\prt\ln(1+\bar w_n)}{\prt\xi}.  \label{eq:47} 
\end{equation}

\noindent We introduce $U_n \equiv \tau_{n+1}^J/\tau_n^J$ and $W_n \equiv \tau_{n-1}^J/\tau_n^J$. Then, with the help of relations (\ref{eq:30}), (\ref{eq:32}) from previous section, (\ref{eq:46}) and (\ref{eq:47}), respectively, can be written as (recall that $v_+ = v_{n+1}-v_n-a_n$, $v_- = v_{n-1}-v_n+a_{n-1}$)

\begin{equation}
\frac{\prt(U_nv_+)}{\prt\xi} = -\xi \frac{\prt U_n}{\prt\xi} - 2\frac{\prt v_n}{\prt\xi}U_n,  \label{eq:48} 
\end{equation}

\begin{equation}
\frac{\prt(W_nv_-)}{\prt\xi} = \xi \frac{\prt W_n}{\prt\xi} - 2\frac{\prt v_n}{\prt\xi}W_n.  \label{eq:49} 
\end{equation}

\noindent Since, by (\ref{eq:41}), $v_+ = -\prt_{t1}\ln U_n$, $v_- = -\prt_{t1}\ln W_n$, the last two equations are equivalent to \ref{eq:53} and \ref{eq:54}.
\end{proof}

\par Thus, on the one hand, there always is a universal system of PDE valid for any UE and very similar in structure to the system derived in Ref.~\cite{IR1} for the Gaussian case using the first equations of 1-Toda lattice hierarchy~\cite{AvM2, AvM7, PvM2007} -- the coupled Toda-AKNS system (see, e.g., Ref.~\cite{Newell}). It consists of the Toda equation itself: 

$$
\frac{\prt^2 \ln\tau_n^J}{\prt t_1^2} = -\frac{\prt v_n}{\prt t_1} + \frac{\prt^2 \ln\tau_n}{\prt t_1^2} = U_nW_n,  \eqno(\ref{eq:52})  
$$

\noindent and the analog of AKNS derived above -- the equations (\ref{eq:53}), (\ref{eq:54}). 
\par However, another important universal equation follows from matching with TW approach, if we recall formula$^{29}$ for the diagonal element of the resolvent kernel taken at an endpoint $\xi$, 

\begin{equation}
R_n^2(\xi,\xi) = p_n\frac{\prt q_n}{\prt \xi} - q_n\frac{\prt p_n}{\prt \xi}.  \label{eq:50} 
\end{equation}

\begin{theorem}

With respect to any particular endpoint $\xi$ of the spectrum for a finite size UE, the following universal PDE holds for matrix integrals -- $\tau$-functions:

$$
\left(\frac{\prt\ln\tau_n^J}{\prt\xi}\right)^2 = R_n^2(\xi,\xi) = -\frac{1}{4}\frac{\prt U_n}{\prt\xi}\frac{\prt W_n}{\prt \xi}\left(\frac{\prt}{\prt \xi}\ln\left(-\frac{\prt U_n/\prt \xi}{\prt W_n/\prt \xi}\right)\right)^2.  \eqno(\ref{eq:51})  
$$

\end{theorem}

\begin{proof}

If we substitute (\ref{eq:50}) into the square of (\ref{eq:25}), use TW formulas (\ref{eq:28}), (\ref{eq:31}), and then formulas (\ref{eq:30}) and (\ref{eq:32}) from previous section, we arrive at equation (\ref{eq:51}).

\end{proof}

This will make it possible to derive something even more interesting -- a universal PDE for only one dependent variable, the logarithm of the $\tau$-function $\tau_n^J$. 
From the TW equations alone one can immediately get the simple bilinear relation:

\begin{equation}
\left(\frac{\prt v_n}{\prt \xi}\right)^2 = \frac{\prt u_n}{\prt \xi} \frac{\prt w_n}{\prt \xi}  \label{eq:55a}  
\end{equation}

\noindent From now on we will denote $T = \ln\tau_n^J$ and write $U$, $W$ instead of $U_n$, $W_n$ and just $t$ instead of $t_1$; also we will write just $f_{\xi}$ for $\prt f/\prt\xi$, $f_t$ for $\prt f/\prt t$ etc. Using the universal relations (\ref{eq:30}), (\ref{eq:32}) and (\ref{eq:41}) with $\tau$-functions, (\ref{eq:55a}) can be rewritten as  

$$
\left(v_{\xi}\right)^2 = \left(T_{\xi t}\right)^2 = -U_{\xi}W_{\xi}.  \eqno(\ref{eq:55})
$$

\noindent The last equation is also a consequence of the above (\ref{eq:53}), (\ref{eq:54}). Thus, (\ref{eq:53}) and (\ref{eq:54}) can be alternatively rewritten as a different pair of equations -- the  equation (\ref{eq:55}) and their other independent combination:

$$
WU_{\xi t} - UW_{\xi t} = \xi\prt_{\xi}(UW) = \xi T_{\xi tt}. \eqno(\ref{eq:56})  
$$

\noindent The equation (\ref{eq:51}) also can be transformed with the help of universal relations into another form:

$$
\prt_t(T_{\xi})^2 = U_{\xi}W_{\xi\xi} - W_{\xi}U_{\xi\xi}, \eqno(\ref{eq:57})  
$$

\noindent which will be used in what follows. It is convenient to introduce functions $G$ and $H$,

\begin{equation}
G = WU_{\xi} - UW_{\xi}, \label{eq:58}  
\end{equation}

\begin{equation}
H = WU_t - UW_t. \label{eq:59} 
\end{equation}

\noindent We now proceed to prove

\begin{theorem}
The logarithm of gap probability $\Pr$ ($T = \ln\tau_n^J = \ln\Pr + \ln\tau_n$) for random matrix unitary ensembles (UE) satisfies a {\it universal PDE}:

$$
\left(T_{\xi t} T_{\xi\xi tt} - T_{\xi tt}T_{\xi\xi t} + 2(T_{\xi t})^3\right)^2 = (T_{\xi})^2\left(4T_{tt}(T_{\xi t})^2 + (T_{\xi tt})^2\right).  \eqno(\ref{eq:T}) 
$$

\end{theorem}

\noindent It would be interesting to find out whether this equation appeared before somewhere in the studies of nonlinear integrable PDE. So far the author is not aware of this.

\begin{proof}

In terms of $G$ (\ref{eq:55}) takes the form (using the Toda equation to eliminate $U$ and $W$):

\begin{equation}
G^2 = 4T_{tt}(T_{\xi t})^2 + (T_{\xi tt})^2. \label{eq:61} 
\end{equation}

\par Now turn to the equation (\ref{eq:57}). One can express the combinations of derivatives of $U$ and $W$ in terms of $T$ and $G$. Using the $\xi$-derivative of equation (\ref{eq:52}) and the definition (\ref{eq:58}), we get

\begin{equation}
2WU_{\xi} = T_{\xi tt} + G,  \label{eq:63}  
\end{equation}

\begin{equation}
2UW_{\xi} = T_{\xi tt} - G.  \label{eq:64}  
\end{equation}

\noindent Differentiating (\ref{eq:63}) and (\ref{eq:64}) w.r.t. $\xi$ and using (\ref{eq:55}) gives, respectively, 

\begin{equation}
2WU_{\xi\xi} = 2(T_{\xi t})^2 + T_{\xi\xi tt} + G_{\xi}, \label{eq:65}  
\end{equation}

\begin{equation}
2UW_{\xi\xi} = 2(T_{\xi t})^2 + T_{\xi\xi tt} - G_{\xi}.  \label{eq:66} 
\end{equation}

\noindent We now multiply (\ref{eq:64}) by (\ref{eq:65}) and subtract the product of (\ref{eq:63}) and (\ref{eq:66}) to get

$$
2UW(W_{\xi}U_{\xi\xi} - U_{\xi}W_{\xi\xi}) = T_{\xi tt}G_{\xi} - (2(T_{\xi t})^2 + T_{\xi\xi tt})G,
$$

\noindent or, using (\ref{eq:52}) and (\ref{eq:57}),

\begin{equation}
T_{\xi tt}G_{\xi} - (2(T_{\xi t})^2 + T_{\xi\xi tt})G = -2T_{tt}\prt_t(T_{\xi})^2.  \label{eq:67} 
\end{equation}

\noindent Differentiating (\ref{eq:61}) w.r.t. $\xi$ one obtains

\begin{equation}
GG_{\xi} = T_{\xi tt}(2(T_{\xi t})^2 + T_{\xi\xi tt}) + 4T_{tt}T_{\xi t}T_{\xi\xi t}.  \label{eq:68} 
\end{equation}


\noindent After multiplying (\ref{eq:67}) by $G$ and applying expressions (\ref{eq:61}) for $G^2$ and (\ref{eq:68}) for $GG_{\xi}$ on the left-hand side some of the terms and common factors cancel out and finally we arrive at

\begin{equation}
T_{\xi}G = T_{\xi t} T_{\xi\xi tt} - T_{\xi tt}T_{\xi\xi t} + 2(T_{\xi t})^3.  \label{eq:69} 
\end{equation}

Thus, we have two independent expressions for the function $G$ in terms of the derivatives of $T$ -- equations (\ref{eq:69}) and (\ref{eq:61}). Substituting G from (\ref{eq:69}) into (\ref{eq:61}) we obtain (\ref{eq:T}). 
\end{proof}

\par Equation (\ref{eq:56}) was not used to obtain (\ref{eq:T}). It contains the $\xi$-variable in the coefficient, while (\ref{eq:T}) has the feature that its coefficients are constant. It turns out that the additional information from (\ref{eq:56}) can be used to obtain another universal PDE, this time for the logarithm of the ratio $T_+ \equiv \ln(\tau_{n+1}^J/\tau_n^J)$ (now `prime' stands for the derivative w.r.t. $\xi$):

\begin{theorem}

The logarithm of the ratio of matrix integrals $T_+ \equiv \ln(\tau_{n+1}^J/\tau_n^J)$ for random matrix unitary ensembles (UE) satisfies a {\it universal PDE}:

$$
(T_+'')_{tt} - \frac{(T_+')_{tt}T_+''}{T_+'} - (T_+')^2(T_+)_{tt} - \frac{(T_+')_{t}(T_+'')_t}{T_+'} + \frac{(T_+')_{t}^2T_+''}{(T_+')^2} + (\xi - (T_+)_t)T_+'(2(T_+')_t - 1) = 0.    \eqno(\ref{eq:T+}) 
$$

\end{theorem}

\begin{proof}
It is a tedious calculation. First one substitutes $W$ expressed from Toda equation (\ref{eq:52}) and $W_{\xi}$ from (\ref{eq:55}) into the $\xi$-derivative of (\ref{eq:52}). This gives

$$
T_{\xi tt} = T_{tt}T_+' - \frac{T_{\xi t}^2}{T_+'}.
$$

\noindent Besides, we use (\ref{eq:53}), writing it as

$$
2T_{\xi t} = -(T_+')_t + (\xi - (T_+)_t)T_+',
$$

\noindent and its $t$-derivative,

$$
2T_{\xi tt} = -(T_+')_{tt} + (\xi - (T_+)_t)(T_+')_t - T_+'(T_+)_{tt}.
$$

\noindent We substitute the last two equations into the previous one, divide it by $T_+'/4$ and get:

$$
4T_{tt} = -2\frac{(T_+')_{tt}}{T_+'} - 2(T_+)_{tt} + \left(\frac{(T_+')_t}{T_+'}\right)^2 + (\xi - (T_+)_t)^2.
$$

\noindent We differentiate the last equation once again w.r.t. $\xi$ and compare the two expressions for $T_{\xi tt}$ in terms of derivatives of $T_+$. This gives (\ref{eq:T+}).
\end{proof}  

\par In addition to the main results -- universal equations (\ref{eq:57})--(\ref{eq:56}) and especially (\ref{eq:T}) and (\ref{eq:T+}), it seems worth to reformulate the preceding findings in this section. In terms of $G$ and $H$ and the main function $T$ (\ref{eq:56}) acquires a nice linear form:

\begin{equation}
G_t + H_{\xi} = 2\xi T_{\xi tt}. \label{eq:60} 
\end{equation}

\noindent Besides, if we use the three equations (\ref{eq:52})--(\ref{eq:54}), eliminating the mixed derivatives $U_{\xi t}$ and $W_{\xi t}$ by substituting their expressions from (\ref{eq:53}) and (\ref{eq:54}), respectively, into (\ref{eq:52}), differentiated twice, once w.r.t. $\xi$ and once w.r.t. $t$, we get (using also $v_{\xi} = -T_{\xi t}$)

$$
T_{\xi ttt} = \xi G + U_{\xi}W_t + W_{\xi}U_t  - 4T_{tt}T_{\xi t}.
$$

\noindent This in turn can be rewritten as an expression of $H$ in terms of $G$ and derivatives of $T$:

\begin{equation}
H = 2\xi T_{tt} - \frac{2T_{tt}T_{\xi ttt} - T_{ttt}T_{\xi tt} - 8(T_{tt})^2T_{\xi t}}{G}.  \label{eq:62} 
\end{equation}

\noindent We thus have general universal expressions for logarithmic derivatives of Toda $\tau$-ratios $\tau_{n \pm 1}/\tau_n$ in terms of derivatives of the main ($n$-level) $\tau$-function (due to explicit expressions (\ref{eq:61}), (\ref{eq:69}) for $G$ and (\ref{eq:62}) for $H$):

$$
2\prt_{\xi}\ln U = \frac{T_{\xi tt} + G}{T_{tt}}, \hspace{2cm} 2\prt_{\xi}\ln W = \frac{T_{\xi tt} - G}{T_{tt}},
$$

$$
2\prt_{t}\ln U = \frac{T_{ttt} + H}{T_{tt}}, \hspace{2cm} 2\prt_{t}\ln W = \frac{T_{ttt} - H}{T_{tt}}.
$$

Therefore we can write out two pairs of general recursion relations for the derivatives of UE $\tau$-functions (restoring the subscript $n$ for the matrix size)

$$
2\prt_{\xi}T_{n+1} = 2\prt_{\xi}T_n + \prt_{\xi}\ln \prt^2_{tt}T_n + \frac{G_n}{\prt^2_{tt}T_n},
$$

$$
2\prt_{\xi}T_n = 2\prt_{\xi}T_{n+1} + \prt_{\xi}\ln \prt^2_{tt}T_{n+1} - \frac{G_{n+1}}{\prt^2_{tt}T_{n+1}},
$$

\noindent and

$$
2\prt_{t}T_{n+1} = 2\prt_{t}T_n + \prt_{t}\ln \prt^2_{tt}T_n + \frac{H_n}{\prt^2_{tt}T_n},
$$

$$
2\prt_{t}T_n = 2\prt_{t}T_{n+1} + \prt_{t}\ln \prt^2_{tt}T_{n+1} - \frac{H_{n+1}}{\prt^2_{tt}T_{n+1}},
$$

\noindent where $G_n \equiv G$, $H_n \equiv H$ at level (matrix size) $n$, and $G_{n+1}$, $H_{n+1}$ are given by the same formulas (\ref{eq:69}), (\ref{eq:62}), only in terms of $\ln\tau_{n+1}^J$ rather than $\ln\tau_n^J$ . The first equations in the pairs follow from previous formulas for $U$, while the second ones -- from the companion formulas for $W$ with the shift $n-1 \to n$.

\section*{\normalsize\bf VI. ISOMONODROMIC DEFORMATIONS, SCHLESINGER EQUATIONS AND TODA-AKNS STRUCTURE} 

The direct connection of TW equations for random matrix UE with general Schlesinger equations -- the compatibility conditions for systems of linear ordinary differential equations (ODE) related to their isomonodromic deformations -- was first pointed out in the Ref.~\cite{HarTW}. There the matrices of the quadratic combinations of TW variables $Q(a_k)$, $P(a_k)$ ($a_k$ is the $k$-th endpoint of the spectrum) were introduced,

\begin{equation}
A_k = (-1)^{k-1}\left(\begin{array}{cc} Q(a_k)P(a_k) & -Q(a_k)^2 \\ P(a_k)^2 & -Q(a_k)P(a_k) \end{array}\right), \label{eq:72} 
\end{equation}

\noindent and shown to satisfy the Schlesinger equations. A little later an important work by Palmer~\cite{Pal} appeared, where the isomonodromic deformations for UE were studied in details. Palmer considered the corresponding ``Cauchy-Riemann operator" $\prt_z$ problem similar to a Riemann-Hilbert problem (RHP) (for different or more general considerations of RHP in the context of matrix UE see e.g. an excellent book~\cite{Dei} as well as the original papers~\cite{FIK1, FIK2, DeiItsZhu, BorDei, DKLVZ}). He studied the deformation of a fundamental matrix $F(z)$,

$$
F(z) = \left(\begin{array}{cc} \varphi & \tilde\varphi \\ \psi & \tilde\psi \end{array}\right),
$$

\noindent of solutions to $2\times 2$ linear system of ODE -- called ``differentiation formulas" in Ref.~\cite{TW1} (for another, more recent, study of general matrix UE based on deformations of such differential systems see Ref.~\cite{BeEyHa-06}):

$$
\frac{d}{dz}\left(\begin{array}{cc} \varphi \\ \psi \end{array}\right) = M(z)\left(\begin{array}{cc} \varphi \\ \psi \end{array}\right),
$$

\begin{equation}
M(z) = \frac{1}{m(z)}\left(\begin{array}{cc} A(z) & B(z) \\ C(z) & -A(z) \end{array}\right),  \label{eq:73} 
\end{equation}

\noindent where $m(z)$ is a polynomial in $z$, the functions $A(z)$, $B(z)$ and $C(z)$ are determined~\cite{Bauldry, TW1}, by the orthogonal functions $\varphi$, $\psi$ (see (\ref{eq:2})) and the corresponding potential $V(x)$:

$$
A(z) = -\int \varphi(y)\psi(y)\cdot \frac{V'(z)-V'(y)}{z-y}dy - \frac{V'(z)}{2},
$$

$$
B(z) = \int \varphi^2(y)\cdot \frac{V'(z)-V'(y)}{z-y}dy,
$$

$$
C(z) = \int \psi^2(y)\cdot \frac{V'(z)-V'(y)}{z-y}dy .
$$

The matrix $F(z)$, whose determinant must be a constant, is normalized by the condition $\det F(z) \equiv 1$. The deformed matrix $Y(z; a)$ is uniquely determined by analyticity requirement outside the real axis, the asymptotics near infinity, $Y(z; a) \sim F(z), z\to\infty$, and matching conditions on a finite number of cuts $[a_j, a_{j+1}]$ on $\Re$:

\begin{equation}
Y_+(z) = Y_-(z)\cdot(I - \chi_{[a_j,a_{j+1}]} + \chi_{[a_j,a_{j+1}]}\Theta),\ \ \ \Theta =  \left(\begin{array}{cc} 1 & 2\pi i\lambda \\ 0 & 1 \end{array}\right),   \label{eq:74} 
\end{equation}

\noindent where $Y_+(z)$ and $Y_-(z)$ are analytic in the upper and lower half-plane, respectively, and define $Y(z)$ in their domains of analyticity. As shown in Ref.~\cite{Pal}, the solution to the above problem can be expressed by the formula

\begin{equation}
Y = \left(\begin{array}{cc} Q & \tilde\varphi - \lambda\tilde KQ \\ P & \tilde\psi - \lambda\tilde KP \end{array}\right),  \label{eq:75} 
\end{equation}

\noindent where $Q$ and $P$ are the functions in (\ref{eq:4}), evaluated for the rescaled operator $\lambda K^J$, i.e. $Q = (I-\lambda K^J)^{-1}\varphi$, $P = (I-\lambda K^J)^{-1}\psi$, and $\tilde K$ is the operator with {\it singular kernel} (i.e. turning to $\infty$ as $x \to y$) 

$$
\tilde K(x,y) = \frac{\tilde\psi(x)\varphi(y) - \tilde\varphi(x)\psi(y)}{x-y}\chi_{J}(y).
$$

\noindent The normalization $\det F = 1$ implies, by analyticity of $\det Y(z)$ and their asymptotic matching at infinity, the corresponding condition for $Y$, $\det Y(z; a) \equiv 1$. This leads to a nontrivilal relation for the elements of $Y$:

\begin{equation}
\det Y = \tilde\psi Q - \tilde\varphi P + \lambda(P\cdot\tilde K Q - Q\cdot\tilde K P) = 1.  \label{eq:76} 
\end{equation}

Then, following Palmer, introduce a matrix function $\Delta(z)$ such that, by definition,

\begin{equation}
F(z)Y_+^{-1}(z) = I + \Delta(z).  \label{eq:77} 
\end{equation}
 
\noindent Differentiation of (\ref{eq:77}) w.r.t. $z$ combined again with analyticity considerations then yields the following ODE satisfied by $Y$:

\begin{equation}
\frac{dY}{dz}Y^{-1} = (I+\Delta)^{-1}M(I+\Delta) - (I+\Delta)^{-1}\frac{d\Delta}{dz}, \label{eq:78} 
\end{equation}

\noindent where $M = M(z)$ is given by (\ref{eq:73}).
The matrix $\Delta$ can be explicitly expressed~\cite{IIKS, Pal} in terms of related functions $\varphi$, $\psi$, $Q$ and $P$. It follows from (\ref{eq:75}), (\ref{eq:76}) and the definitions (\ref{eq:77}) and (\ref{eq:74}), that

\begin{equation}
\Delta(z; J) = \intop \frac{\lambda\chi_J(y)}{z-y}dy\left(\begin{array}{cc} -\varphi(y)P(y) & \varphi(y)Q(y) \\ \\ -\psi(y)P(y) & \psi(y)Q(y) \end{array}\right).  \label{eq:79} 
\end{equation}

Consideration of monodromies around the ends of the branch cuts $a_j$ for $Y$ and the singular point at $\infty$ leads to the conclusion~\cite{Pal} that there exists a matrix function $\Pi(z; J)$ analytic everywhere in $z$-plane except for the point at $\infty$ and possible zeros of $m(z)$ (see (\ref{eq:73}), (\ref{eq:78})) such that

\begin{equation}
\frac{dY}{dz}Y^{-1} = \Pi(z; J) + \lambda \sum_k \frac{A_k}{z-a_k},  \label{eq:80} 
\end{equation}

\noindent with matrices $A_k$ given by (\ref{eq:72}). In the simpler case, when $m(z) \equiv 1$, the function $\Pi(z)$ is polynomial in $z$. One can see from general equations (\ref{eq:78}) and (\ref{eq:80}) that the principal part at $\infty$,

\begin{equation}
\text{p. p. }(Y'Y^{-1}) = \text{p. p. }(\Pi) = \text{p. p. }((I+\Delta)^{-1}M(I+\Delta)). \label{eq:81}  
\end{equation}

\noindent Also in general $m(z)\Pi(z)$ is a matrix polynomial having no singularities besides $\infty$ so that

\begin{equation}
\text{p. p. }(m(z)\Pi(z)) = m(z)\Pi(z).  \label{eq:82} 
\end{equation}

\noindent This in principle determines matrix $\Pi$ in general from $M$ and $\Delta$ given by (\ref{eq:73}) and (\ref{eq:79}).
\par The equations for the derivatives of $Y$ w.r.t. the endpoints $a_j$ follow also from (\ref{eq:77}):

$$
-FY^{-1}\prt_aYY^{-1} = \prt_a\Delta(z; J)
$$

\noindent or

$$
\prt_aYY^{-1} = - (I+\Delta)^{-1}\prt_a\Delta(z; J) = O(z^{-1}) \ \text{  as } z\to \infty,
$$

\noindent due to the asymptotic expansion as $z \to \infty$,

\begin{equation}
\Delta(z) = \frac{\Delta_1}{z} + \frac{\Delta_2}{z^2} + \dots   \label{eq:83} 
\end{equation}

\noindent The last equations determine that near the endpoints $a_j$

$$
\text{p. p. }(\prt_aYY^{-1}) = -\frac{\lambda A_j}{z-a_j},
$$

\noindent then by analyticity of $\prt_aYY^{-1}$ in $z$ except for the poles at $a_j$, one has

\begin{equation}
\prt_aYY^{-1} = -\frac{\lambda A_j}{z-a_j}.  \label{eq:84} 
\end{equation}

\noindent Calculating the mixed derivatives $\prt^2_{za_j}Y$ and $\prt^2_{a_ka_j}Y$ from equations (\ref{eq:80}) and (\ref{eq:84}) for different $j$ and equating the {\it singular parts} when considering various limits $z \to a_j$, gives Schlesinger equations modified~\cite{Pal} by $\Pi$, see also Ref.~\cite{BorDei}:

\begin{equation}
\frac{\prt A_j}{\prt a_j} = [\Pi(a_j; J), A_j] - \lambda\sum_{k\ne j} \frac{[A_j, A_k]}{a_j-a_k},  \label{eq:85} 
\end{equation}

\begin{equation}
\frac{\prt A_j}{\prt a_k} = \lambda\frac{[A_j, A_k]}{a_j-a_k} = \frac{\prt A_k}{\prt a_j}.  \label{eq:86}  
\end{equation}

\noindent Another nonstandard Schlesinger equation for $\Pi$ arises from the same calculation~\cite{Pal}:

\begin{equation}
\frac{\prt \Pi(z; J)}{\prt a_j} = \lambda\frac{[\Pi(z; J) - \Pi(a_j; J), A_j]}{z-a_j}.  \label{eq:87} 
\end{equation}

\noindent If we consider the limit $z \to \infty$ in the last equation as Palmer did and write $M(z) = M_0(z)(I + o(1))$ near infinity, then, using (\ref{eq:81}) and the expansion (\ref{eq:83}), we get

$$
[M_0, \frac{\prt\Delta_1}{\prt a_j} - \lambda A_j] = 0.
$$

\noindent Since matrix $M_0$ is arbitrary, a completely universal equation follows:

\begin{equation}
\frac{\prt\Delta_1}{\prt a_j} = \lambda A_j,  \label{eq:88} 
\end{equation}

\noindent with $\Delta_1$ determined by expansion of the formula (\ref{eq:79}):

\begin{equation}
\Delta_1 = \lambda\left(\begin{array}{cc} -v & u \\ -w & v \end{array}\right).   \label{eq:89} 
\end{equation}

\noindent We see that parameter $\lambda$ drops out of (88) and we obtain exactly all three universal TW equations (\ref{eq:28}), (\ref{eq:31}) and (\ref{eq:33}) from Ref.~\cite{TW1} for the ``auxiliary" variables $u$, $v$ and $w$, which now are given the special importance and meaning in terms of $\tau$-functions, see (\ref{eq:30}), (\ref{eq:32}) and (\ref{eq:41a}). There is also a general expression for logarithmic derivatives of $\tau_n^J$ w.r.t. the endpoints in terms of $A_j$ and $\Pi$ in the isomonodromic approach (e.g. in Refs.~\cite{Pal, BorDei}):

\begin{equation}
\frac{\prt \ln\tau_n^J}{\prt a_j} = \text{Tr} \left[\Pi(a_j; J)A_j + \sum_{k\neq j}\frac{A_kA_j}{a_j-a_k}\right],  \label{eq:90} 
\end{equation}

\noindent which is a consequence of existence of fundamental closed form $\Omega$, found first for the most general case of systems of linear ODE with rational coefficients in Ref.~\cite{JMU}:

\begin{equation}
\Omega = d_a \ln\tau^J = d_a \ln\det(I-K^J) = \sum_j \text{Tr} \left[\Pi(a_j; J)A_j + \sum_{k\neq j}\frac{A_kA_j}{a_j-a_k}\right]da_j.   \label{eq:91} 
\end{equation}

\par If we now sum up the equations (\ref{eq:85}) and (\ref{eq:86}) over {\it all} endpoints $a_j$ {\it and} over {\it all} matrices $A_j$ and use the equations (\ref{eq:88}), we arrive at a system of three second order PDE for three functions $u$, $v$ and $w$ (the two equations for diagonal components are the same due to the tracelessness of the matrices involved):

\begin{equation}
\B_{-1}^2 u = -\sum_j [\Pi(a_j; J), A_j]_{12},  \label{eq:92} 
\end{equation}

\begin{equation}
\B_{-1}^2 w = \sum_j [\Pi(a_j; J), A_j]_{21}, \label{eq:93} 
\end{equation}

\begin{equation}
\B_{-1}^2 v = -\sum_j [\Pi(a_j; J), A_j]_{11}.  \label{eq:94} 
\end{equation}

\noindent Rewritten in terms of $\tau$-functions with the help of the universal relations (\ref{eq:30}), (\ref{eq:32}), and (\ref{eq:41}) this becomes the universal analogue of the system of ASvM type obtained in Ref.~\cite{IR1} for the Gaussian case. Note, however, that the equation (\ref{eq:94}) becomes in fact the {\it third} order PDE after using (\ref{eq:41}). For the Gaussian case, though, it turns out to be the total derivative of the first integral obtained in Ref.~\cite{TW1}. And the first integral itself is identified by (\ref{eq:34}) with our ``boundary-Toda" equation~\cite{IR1}.
\par Let us now sum up only equations (\ref{eq:85}) and (\ref{eq:86}) with derivatives w.r.t. a {\it single particular endpoint} $a_j$ only, using again (\ref{eq:88}) and (\ref{eq:89}) so that 

\begin{equation}
\sum_k A_k = \frac{1}{\lambda}\sum_k  \frac{\prt\Delta_1}{\prt a_k} = \B_{-1}\left(\begin{array}{cc} -v & u \\ -w & v \end{array}\right).  \label{eq:95} 
\end{equation}

\noindent Again only terms involving $\Pi$ remain on the right-hand side. Expanding these commutators we write down 

\begin{prop}
The three independent components of the Schlesinger equations associated with UE, summed over all $a_k$ for every matrix $A_j$, have the form of three-term PDE, 

\begin{equation}
\frac{\prt}{\prt a_j} \B_{-1} u = (\Pi_0)_j\frac{\prt u}{\prt a_j} + 2(\Pi_+)_j\frac{\prt v}{\prt a_j},   \label{eq:96} 
\end{equation}

\begin{equation}
\frac{\prt}{\prt a_j} \B_{-1} w = -(\Pi_0)_j\frac{\prt w}{\prt a_j} + 2(\Pi_-)_j\frac{\prt v}{\prt a_j},   \label{eq:97} 
\end{equation}

\begin{equation}
\frac{\prt}{\prt a_j} \B_{-1} v = (\Pi_-)_j\frac{\prt u}{\prt a_j} + (\Pi_+)_j\frac{\prt w}{\prt a_j},   \label{eq:98} 
\end{equation}

\noindent where we denoted $\Pi_j = \Pi(a_j; J)$ and matrix elements of $\Pi$ as $\Pi_+ = \Pi_{12}$, $\Pi_- = \Pi_{21}$, $\Pi_0 = \Pi_{11} - \Pi_{22} = 2\Pi_{11}$.

\end{prop}

\noindent The last equality is again due to $\Pi$ being traceless as a consequence of tracelessness of the original matrix $M(z)$ of differential relations. 
\par One can see a striking similarity in structure of the last equations and equations (\ref{eq:52})--(\ref{eq:54})  from previous section. Only the current equation (\ref{eq:98}) looks like a derivative of the Toda equation (\ref{eq:52}), which it is indeed in the special Gaussian case. The Gaussian case is now seen as a very special in the respect that there the general structural correspondence of the two systems holds literally term by term. This is not true for the other UE but still the Toda-AKNS-like structure and three-term character are present in both forms {\it as if such exact correspondence existed at some deeper level}. We believe this must be actually true although the real explanation is lacking yet.
\par Let us now consider further the simplest (though important) case of only one endpoint $\xi \equiv a_1$. Then $\B_{-1} = \prt_{\xi}$ and we rewrite the system (\ref{eq:98}), (\ref{eq:96}), (\ref{eq:97}) applying also the universal expressions of $u$ and $w$ in terms of $\tau$-ratios. Then, if we introduce

$$
\tilde\Pi_+ = \frac{\tau_{n+1}}{\tau_n b_{n-1}}\Pi_+, \hspace{2cm}  \tilde\Pi_- = \frac{\tau_{n-1}}{\tau_n b_{n-1}}\Pi_-,  
$$

\noindent our system reads:

\begin{equation}
U_{\xi\xi} = \Pi_0U_{\xi} - 2\tilde\Pi_+v_{\xi},   \label{eq:100} 
\end{equation}

\begin{equation}
W_{\xi\xi} = -\Pi_0W_{\xi} + 2\tilde\Pi_-v_{\xi},   \label{eq:101} 
\end{equation}

\begin{equation}
v_{\xi\xi} = -\tilde\Pi_-U_{\xi} + \tilde\Pi_+W_{\xi}.   \label{eq:102} 
\end{equation}

\noindent We can add to this system equation (\ref{eq:90}), which for the single endpoint case can be written as

$$
T_{\xi} = -\Pi_0v_{\xi} + \Pi_-u_{\xi} - \Pi_+w_{\xi} =  -\Pi_0v_{\xi} - \tilde\Pi_-U_{\xi} - \tilde\Pi_+W_{\xi}.  \eqno(\ref{eq:90}a) 
$$

\noindent Here so far there are only $\xi$-derivatives and no $t$-derivatives, however, the elements of $\Pi$, although known in principle from Palmer's considerations (see (\ref{eq:81}) and (\ref{eq:82})), still do not have nice explicit expressions in general. Therefore, in general we obtain only some {\it universal constraints} for the non-universal $\Pi$-quantities. The system (\ref{eq:100})--(\ref{eq:102}), if treated as a linear algebraic system w.r.t. $\Pi_0$, $\tilde\Pi_+$ and $\tilde\Pi_-$, is degenerate, and equation (\ref{eq:55}), $(v_{\xi})^2 = -U_{\xi}W_{\xi}$, is also a consequence of (\ref{eq:100}), (\ref{eq:101}) like it was for (\ref{eq:53}), (\ref{eq:54}). The other independent combination of (\ref{eq:100}) and (\ref{eq:101}) together with (\ref{eq:90}a) bring back our universal equation (\ref{eq:51}):

$$
W_{\xi}U_{\xi\xi} - U_{\xi}W_{\xi\xi} = 2v_{\xi}(-\Pi_0v_{\xi} - \tilde\Pi_-U_{\xi} - \tilde\Pi_+W_{\xi}) = 2v_{\xi}T_{\xi}.   
$$

\noindent Thus one can get only two rather than three universal constraints for $\Pi$-quantities, i.e. their expressions in terms of the derivatives of $T \equiv \ln\tau_n^J$. They are in fact equation (\ref{eq:90}a) and also (\ref{eq:102}) itself. 


\section*{\normalsize\bf VII. UNIVERSAL PDE (\ref{eq:T}) AND EXAMPLES} 

{\it Example 1: Gaussian emsemble.} \\
The Virasoro constraints~\cite{ASvM},~\cite{AvM7} give in this case (consider one interval situation -- the distribution of largest eigenvalue):

$$
\prt_t T = -\frac{1}{2}\prt_{\xi} T, \ \ \ \prt_{tt}^2 T = \frac{1}{4}\prt_{\xi\xi}^2 T + \frac{n}{2}.
$$

\noindent We substitute these expressions for the derivatives of $\ln\tau_n^J$ w.r.t. the first Toda time into our universal PDE (\ref{eq:T}) and, denoting $T' = \prt_{\xi}T$, get a fourth-order ODE:

\begin{equation}
\left(T''T'''' - (T''')^2 + 2(T'')^3\right)^2 = 4(T')^2\left((T''')^2 + 4(T'')^2(T'' + 2n)\right). \label{eq:105} 
\end{equation}

\noindent This is actually a third order ODE for $r \equiv T'$:

\begin{equation}
\left(r'r''' - (r'')^2 + 2(r')^3\right)^2 = 4r^2\left((r'')^2 + 4(r')^2(r' + 2n)\right). \label{eq:106} 
\end{equation}

\noindent This third-order ODE is a nonlinear equation with coefficients not involving the independent variable $\xi$, unlike the third-order equation for Gaussian ensemble~\cite{TW1}:

\begin{equation}
r''' + 6(r')^2 + 8nr' = 4\xi(\xi r' - r), \label{eq:107} 
\end{equation}

\noindent which can be integrated once to give a form of Painlev\'e IV equation~\cite{TW1}:

\begin{equation}
(r'')^2 + 4(r')^2(r'+2n) = 4(\xi r'-r)^2. \label{eq:108} 
\end{equation}

\noindent One can verify that our equation (\ref{eq:106}) is in fact a consequence of (\ref{eq:108}). The Painlev\'e equation (\ref{eq:108}) itself can also be obtained from (\ref{eq:106}). To this end define a function $g(\xi)$ such that the right-hand side of (\ref{eq:106}) equals $4r^2g^2$ and

\begin{equation}
r'r''' - (r'')^2 + 2(r')^3 = 2rg. \label{eq:109} 
\end{equation}

\noindent Differentiating the expression on the right-hand side of (\ref{eq:106}), defining $g^2(\xi)$ we have also

\begin{equation}
(g^2)' = 2gg' = 2r''r''' + 8r'r''(r'+2n) + 4(r')^2r'' = 2r''(r''' + 6(r')^2 + 8nr').  \label{eq:110} 
\end{equation}

\noindent Eliminating the third derivative $r'''$ from (\ref{eq:109}) and (\ref{eq:110}) gives

$$
r'gg' = r''g^2 + 2rr''g
$$

\noindent or

\begin{equation}
r'g' - r''g = 2rr''.   \label{eq:111} 
\end{equation}

\noindent Substituting $g = r'\Phi$ into (\ref{eq:111}) and solving the ODE for $\Phi$, then using the boundary condition $r \to 0$ as $\xi \to \infty$, one gets

$$
g = 2(\xi r' - r),
$$

\noindent and, finally, using again the definition of $g^2$ from the right-hand side of (\ref{eq:106}), one arrives at (\ref{eq:108}).

\medskip

\noindent {\it Example 2: Laguerre ensemble.} \\
Here the potential is $V(x) = x - \alpha\ln x$, and the relevant Virasoro constraints~\cite{ASvM},~\cite{AvM7} give expressions of Toda-time derivatives in terms of the endpoint ones:

$$
\prt_t T = -\xi T' + n(n+\alpha), \ \ \ \prt_{tt} T = \xi^2T'' + n(n+\alpha),
$$

\noindent so, for the function $r = T'$, one finds

$$
r_t = T_{\xi t} = -\xi r' - r, \ \ \ r_{tt} = T_{\xi tt} = \xi^2r'' + 2\xi r',
$$

$$
r_{\xi t} = T_{\xi\xi t} = -\xi r'' - 2r', \ \ \ r_{\xi tt} = T_{\xi\xi tt} = \xi^2r''' + 4\xi r'' + 2r'.
$$

\noindent For the Laguerre case it is convenient and standard to introduce the function $\sigma = \xi r$ instead of $r$, then one can see that

$$
r_t = -\sigma', \ \ \ r_{tt} = \xi\sigma'', \ \ \ r_{\xi t} = -\sigma'', \ \ \ r_{\xi tt} = \xi\sigma''' + \sigma'',
$$

\noindent and the universal PDE (\ref{eq:T}) reduces to the following ODE for $\sigma$:

\begin{equation}
\xi^2\left(-\sigma'(\xi\sigma''' + \sigma'') + \xi(\sigma'')^2 - 2(\sigma')^3\right)^2 = \sigma^2\left(4(\xi\sigma'-\sigma+n(n+\alpha))(\sigma')^2 + \xi^2(\sigma'')^2\right). \label{eq:112} 
\end{equation}

\noindent Again let us introduce function $g(\xi)$ such that r.h.s. of (\ref{eq:112}) $= \sigma^2g^2$ and

\begin{equation}
\xi\left(\sigma'(\xi\sigma''' + \sigma'') - \xi(\sigma'')^2 + 2(\sigma')^3\right) = \sigma g. \label{eq:113} 
\end{equation}

\noindent Differentiating the expression for $g^2$ on the r.h.s. of (\ref{eq:112}) gives

\begin{equation}
gg' = \sigma''\left(\xi^2\sigma''' + \xi\sigma'' + 6\xi(\sigma')^2 - 4\sigma\sigma' + 4n(n+\alpha)\sigma'\right).  \label{eq:114} 
\end{equation}

\noindent Again combining (\ref{eq:113}) and (\ref{eq:114}) to eliminate the third derivative $\sigma'''$ leads to

$$
\sigma'gg' = \sigma''g(g+\sigma),
$$

\noindent  or

\begin{equation}
\sigma'g' - \sigma''g = gg'', \label{eq:115} 
\end{equation}

\noindent which is the same as (\ref{eq:111}) up to a factor of 2 on the right-hand side. So one gets

\begin{equation}
g = \xi\sigma' - \sigma.  \label{eq:116} 
\end{equation}

\noindent Finally from the r.h.s. of (\ref{eq:112}) and (\ref{eq:116}) we arrive at

\begin{equation}
\xi^2(\sigma'')^2 + 4(\xi\sigma'-\sigma+n(n+\alpha))(\sigma')^2  - (\xi\sigma' - \sigma)^2 = 0,  \label{eq:117} 
\end{equation}

\noindent which is the standard form of Painlev\'e V equation for Laguerre ensemble~\cite{TW1, ASvM}.

\bigskip

The procedure just done for two classical ensembles can be applied for equation (\ref{eq:T}) in general. If we put $r = T_{\xi}$ then the universal PDE (\ref{eq:T}) reads

\begin{equation}
\left(r_t r_{\xi tt} - r_{tt}r_{\xi t} + 2(r_t)^3\right)^2 = r^2\left(4T_{tt}(r_t)^2 + (r_{tt})^2\right).  \label{eq:118} 
\end{equation}

\noindent We introduce function $\Phi$ such that (actually $r_t\Phi = G$, with $G$ defined by (\ref{eq:58}) and satisfying (\ref{eq:61}), (\ref{eq:69}))

\begin{equation}
\Phi^2 = \left(\frac{r_{tt}}{r_t}\right)^2 + 4T_{tt}  \label{eq:119} 
\end{equation}

\noindent and

\begin{equation}
r_t r_{\xi tt} - r_{tt}r_{\xi t} + 2(r_t)^3 = rr_t\Phi.  \label{eq:120} 
\end{equation}

\noindent Now differentiate (\ref{eq:119}) with respect to $\xi$ and eliminate the most senior derivative $r_{\xi tt}$ with the help of (\ref{eq:120}). Then one obtains

\begin{equation}
\Phi_{\xi} = \frac{rr_{tt}}{(r_t)^2}.  \label{eq:121} 
\end{equation}

\noindent Two equations -- (\ref{eq:119}) and (\ref{eq:121}) -- together give an equivalent representation of the single equation (\ref{eq:118}). If (\ref{eq:121}) could be integrated in general to get an expression for $\Phi$, then substituting it into (\ref{eq:119}) would give a universal analog of Painlev\'e equations for all unitary ensembles. Unfortunately, this is not the case. And, most important, there seems to be no not only general method to express the two $t$-derivatives -- first and second -- entering the universal PDE in terms of $\xi$-derivatives and functions of $\xi$ only but even case by case application of Virasoro constraints beyond classical ensembles leads to infinitely many equations connecting derivatives w.r.t. all different times of Toda hierarchy, together with $\xi$-derivatives. This seems to forbid simply eliminating the $t$-derivatives from the universal PDE for non-classical ensembles. To make this equation useful, a different point of view is needed. It might require some considerations of differential systems of the previous section, see also e.g. their treatment in Ref.~\cite{BeEyHa-06}, or a completely different perspective yet to be found. This open problem seems worth futher efforts since the approach presented here seems rather appealing.

\section*{\normalsize\bf VIII. CONCLUSIONS}

The simple universal connections between eigenvalue spacing probabilities and their ratios represented as ratios of one-dimensional Toda $\tau$-functions on the one side, and auxiliary variables naturally arising in the analytic approach to PDE for the probabilities using Fredholm determinants and resolvent operators on the other side, are derived. They hold for all Hermitian random matrix ensembles with unitary symmetry of the probability measure. They are obtained as a consequence of also derived here universal three-term recurrence relations for systems of functions orthogonal on subsets of real line. These orthogonal bases for restricted ensembles were studied by Borodin and Soshnikov in Ref.~\cite{BorSosh}, so we make a missing connection of their work with Ref.~\cite{TW1}. All our considerations can be generalized to unitarily invariant random matrix ensembles with eigenvalues on a circle or other simple contours in complex plane.  
\par The above relations of functional-analytic approach with the structure of orthogonal functions and one-dimensional Toda (or Toda-AKNS) integrable hierarchy allow us to derive some universal PDE for 1-Toda $\tau$-functions-matrix integrals and their ratios. Even a single universal PDE for the logarithm of gap probability (\ref{eq:T}) of various unitary ensembles is obtained. Although the difficulty of applying the universal PDE (\ref{eq:T}) directly to obtain specific ODE for various non-classical ensembles is rather disappointing, the author believes that equation (\ref{eq:T}), and also (\ref{eq:T+}) for the ratio of consecutive size matrix integrals,  may be important and could be somehow used to describe both universal and model-specific properties of random matrix UE. How to use them remains an interesting question for future research.

\bigskip
\bigskip
{\noindent\bf ACKNOWLEDGMENTS} \\
Author is especially grateful to C.A.Tracy for constant support and encouragement and valuable comments on the text of the paper. Special thanks are also due to A.Borodin for his general interest in this work and turning author's attention to paper of Palmer~\cite{Pal}. Author also wishes to thank A.Its and A.Soshnikov for useful discussions; E.Kanzieper for organizing the ISF Research Workshop on Random Matrices and Integrability, Yad Hashmona, Israel, March 2009, where a preliminary version of this work was presented. Author thanks the referee whose comments helped improve the previous version of the manuscript.
\par This work was supported in part by National Science Foundation under grant DMS-0906387 and VIGRE grant DMS-0636297.


\end{document}